\documentclass[journal]{IEEEtran}
\usepackage{epsf}
\usepackage{graphicx}
\usepackage{amsmath,bm}
\usepackage{amssymb}
\usepackage{epsfig,latexsym,amsmath,epsf,amssymb,amsfonts}
\usepackage{verbatim}
\usepackage{placeins}
\usepackage[hang]{subfigure}
\usepackage{epstopdf}
\usepackage{multirow}
\usepackage{cite} 
\usepackage{amsthm}
\usepackage{url}
\usepackage{cases}
\usepackage[utf8]{inputenc}
\usepackage{chngcntr}
\usepackage{tabu}
\usepackage[linesnumbered,ruled]{algorithm2e}
\usepackage[justification=centering]{caption}
\usepackage{siunitx}

\counterwithin{paragraph}{section} 

\newtheorem*{lemma}{Lemma}

\newcommand{\ignore}[1]{}
\usepackage{color, colortbl}

\usepackage[belowskip=-5pt,aboveskip=0pt]{caption}

\begin{document}

\setlength{\abovedisplayskip}{3pt}
\setlength{\belowdisplayskip}{3pt}

\title{\Large 3D Placement of an Unmanned Aerial Vehicle Base Station for Maximum Coverage of Users with Different QoS Requirements}
\author{Mohamed Alzenad, Amr El-Keyi, and~Halim~Yanikomeroglu
\thanks{This work was supported by the Ministry of Higher Education and Scientific Research (MOHESR), Libya,
through the Libyan-North American Scholarship Program, and in part by TELUS Canada.}
\thanks{The authors are with the Department of Systems and Computer Engineering, Carleton University, Ottawa, Ontario, Canada. Email: \{mohamed.alzenad, amr.elkeyi, halim\}@sce.carleton.ca, M. Alzenad is also affiliated to Sirte University, Libya.}
}

\maketitle

\begin{abstract}
The need for a rapid-to-deploy solution for providing wireless cellular services can be realized by unmanned aerial vehicle base stations (UAV-BSs). To the best of our knowledge, this letter is the first in literature that studies a novel 3D UAV-BS placement that maximizes the number of covered users with different Quality-of-Service requirements. We model the placement problem as a multiple circles placement problem and propose an optimal placement algorithm that utilizes an exhaustive search (ES) over a one-dimensional parameter in a closed region. We also propose a low-complexity algorithm, namely, maximal weighted area (MWA) algorithm to tackle the placement problem.
 Numerical simulations are presented showing that the MWA algorithm performs very close to the ES algorithm with a significant complexity reduction.
\end{abstract}

\section{introduction}
The concept of unmanned aerial vehicle base stations (UAV-BSs) has emerged as a rapid solution for providing wireless services \cite{IremMagazine,zeng2016wireless}. 
 The need for UAV-BSs could arise in various scenarios, for instance, during a malfunction of the terrestrial infrastructure or for the purpose of offloading traffic from a congested macro BS \cite{IremMagazine}. UAV-BSs can also play a key role for providing an energy efficient internet of things (IoT) communications where UAV-BSs can collect data from the IoT devices and forward it to other devices \cite{mozaffari2017mobile}. 
 
Despite its promising benefits, UAV-aided communication is facing many challenges. Unlike terrestrial channels, where the location of the BS is fixed, and hence the path loss depends on the location of the user,  the air-to-ground (A2G) channel model is a function of the location of the user as well as the UAV-BS. A key challenge in UAV-aided communications is where to deploy the UAV-BS. Furthermore, the UAV-BS placement is no longer a 2D placement problem as for terrestrial BSs. It is indeed a 3D placement problem. Furthermore, the energy available for powering the onboard electronics is limited because of using batteries as a source of power \cite{zeng2017energy}. Therefore, the UAV-BS may not be capable of providing a full coverage for the serving area, and only partial coverage is possible. A key challenge that is addressed in this letter is that given a limited UAV-BS transmit power and users with different quality of service (QoS) requirements, defined in terms of the received signal to noise ratio (SNR), where to deploy the UAV-BS such that the number of covered users is maximized.      

The work in \cite{lyu2017placement} proposed a polynomial-time spiral algorithm for multiple UAVs placement. The authors in \cite{Alzenad} proposed a framework for evaluating the 3D location of the UAV-BS that maximizes the number of covered users using minimum transmit power. The work in \cite{Hourani} evaluated the optimal UAV-BS altitude that maximizes the coverage region. The work in \cite{IremConf} made a further step and deployed the UAV-BS based on the locations of the users and formulated the UAV-BS placement problem as a quadratically-constraint mixed integer non-linear problem. A grid search algorithm was proposed in \cite{Elham2017Backhaul} to tackle a backhaul-aware 3D UAV-BS placement problem. The authors in \cite{ElhamVtc} developed a particle swarm optimization framework to find the minimum number of UAV-BSs and their locations to serve a particular region. A 3D UAV-BS placement for two cases, one UAV-BS and two UAV-BSs was examined in \cite{MozaffariDrone}. Furthermore, the authors in \cite{MozaffariDrone} optimized the 3D UAV-BS deployment with the aim of maximizing the coverage region with the minimum transmit power. However, the work in \cite{lyu2017placement,Alzenad,Hourani,IremConf,Elham2017Backhaul,ElhamVtc,MozaffariDrone} assumes that all the users have the same QoS requirement.

In this letter, we study a novel 3D UAV-BS placement that has not been previously addressed. 
Our work aims to maximize the number of covered users demanding different QoS requirements. We model the UAV-BS placement as a multiple circles placement problem. 
 We propose an algorithm that utilizes an exhaustive search (ES) over a one-dimensional parameter in a closed region to determine the optimal height and 2D location of the UAV-BS. In addition, we propose a low-complexity algorithm, namely maximal weighted area (MWA) algorithm to solve the placement problem. We also show by simulations that the proposed MWA algorithm performs very close to the ES algorithm with a significant complexity reduction.  
\section{system model}
We consider a congested area  containing a set of stationary or low-mobility users. The congestion  at the terrestrial BS might have occurred due to a number  of reasons including a malfunction at the BS or a temporary event such  as a festival or a sports event. Therefore, in order to relieve the  stress at the terrestrial BS, a UAV-BS is deployed for serving as many  users as possible. We assume that each user has one of $K$ different QoS requirements defined in terms of the  SNR. Let $\mathcal U$ denote the set of the users and $\mathcal U_k\subseteq\mathcal U$ is the set of the users corresponding to QoS $k$ such that $\cup_{k=1}^K \mathcal U_k=\mathcal U $. We also denote by $(x_{ik},y_{ik}),i=1,2,...|\mathcal U_k|,k=1,2,..K$, the 2D location of the user $i$ of the set $\mathcal U_k$.

As discussed in \cite{Hourani}, the A2G links 
are either line-of-sight (LoS) or non line-of-sight (NLoS) with some probability. 
Assuming a UAV-BS located at $(x_D,y_D,h)$, the path loss for the LoS and NLoS links in dB is given respectively by  
\begin{align}\label{Eq:LoSNLoSPL}
  L_\textup{LoS}&= 20\log\left(\frac{4\pi f_c d_{ik}}{c}\right)+\eta_{\textup{LoS}} \nonumber \\
  L_\textup{NLoS}&= 20\log\left(\frac{4\pi f_c d_{ik}}{c}\right)+\eta_{\textup{NLoS}},
\end{align}
where $f_c$ is the carrier frequency, $d_{ik}$ is the distance between the UAV-BS and user $i$ of $\mathcal U_k$, given by $d_{ik}=\sqrt{h^2+r_{ik}^2}$, where $r_{ik}=\sqrt{(x_{ik}-x_D)^2+(y_{ik}-y_D)^2}$. Furthermore, $\eta_{\textup{LoS}}$ and $\eta_{\textup{NLoS}}$ are the average additional losses for LoS and NLoS, respectively, and are given in \cite{Hourani}. The probability of occurrence of a LoS connection between the UAV-BS and user $i$ of set $\mathcal U_k$ located at an elevation angle $\theta_{ik}= \tan^{-1}(\frac{h}{r_{ik}})$ is given by
\begin{equation}\label{Eq:Prob}
\textup P_\textup {LoS}=\frac{1}{1+a \exp (-b(\frac{180}{\pi}\theta_{ik}-a))},
\end{equation} 
where $a$ and $b$ are constants that depend on the environment. Also, the probability of NLoS is $\textup P_\textup{NLoS}=1-\textup P_\textup{LoS}$. In this letter, we only deal with the mean path loss rather than its random behavior because BS deployment often deals with long term variations of the channel rather than small scale variations \cite{Alzenad}. Finally, the probabilistic mean path loss is given by
\begin{equation}\label{Eq:PL1}
L(h,r_{ik})=L_\textup{LoS} P_\textup{LoS}+L_\textup{NLoS} P_\textup{NLoS},
\end{equation}
which yields
\begin{equation}\label{Eq:PL2}
\resizebox{.89\columnwidth}{!}{$\displaystyle{L(h,r_{ik})=\frac{A}{1+a \exp (-b(\frac{180}{\pi}\tan^{-1}(\frac{h}{r_{ik}})-a))}+10\log(h^2+r_{ik}^2)+B}$},
\end{equation}
where $A=\eta_{\textup{LoS}}-\eta_{\textup{NLoS}}$ and $B=20\log(\frac{4\pi f_c}{c})+\eta_{\textup{NLoS}}$. Equation (\ref{Eq:PL2}) can be further rewritten as
\begin{equation}\label{Eq:PL3}
\resizebox{.89\columnwidth}{!}{$\displaystyle{L(h,r_{ik})=\frac{A}{1+a \exp (-b(\frac{180}{\pi}\theta_{ik}-a))}+20\log(\frac{r_{ik}}{\cos (\theta_{ik})})+B}$}.
\end{equation}


Let $P_t$ denote the transmit power of the UAV-BS in dB. The received power at user $i$ of the set $\mathcal U_k$ in dB is given by
\begin{equation}\label{Eq:PowRec}
P_r^{ik}=P_t-L(h,r_{ik}).
\end{equation}
In a noise limited scenario, the conventional approach to define coverage is through the SNR. The $i$th user of set $\mathcal U_k$ is covered if the probabilistic mean SNR exceeds a predefined threshold $\gamma_\textup {th}^k$ (dB). That is if
\begin{equation}
\gamma(h,r_{ik}) \,(\textup {dB})=P_ r^{ik}-P_n=P_t-L(h,r_{ik})-P_n\geq \gamma_\textup {th}^k
\end{equation}
where $P_n$ is the noise power in dB. Clearly, the coverage condition can be equivalently defined in terms of the probabilistic mean path loss. Hence, a user $i$ of set $\mathcal U_k$ is covered if its link experiences a mean path loss less than or equal to some threshold $L_{\textup {th}}^k$, where $L_{\textup {th}}^k=P_t-P_n-\gamma_\textup {th}^k$.  

It was shown in \cite{Alzenad,Hourani} and can also be seen from (\ref{Eq:PL2}) that, for a given environment, a UAV-BS altitude and a QoS requirement $L_{\textup {th}}^k$, the coverage region is a circular disc with radius $R_k(h)=r|_{L(h,r_k)=L_{\textup {th}}^k}$. However, for multiple QoS requirements, the coverage region is no longer a single circular disc. We can see from (\ref{Eq:PL2}) that the region over which all the QoS requirements $\{L(h,r_k)\leq L_{\textup {th}}^k\}_{k=1}^K$ are satisfied forms a set of circular discs with radii $\{R_k(h)\}_{k=1}^K$ and center $(x_D,y_D)$. Obviously, the larger the required path loss threshold $L_{\textup {th}}^k$, the larger the coverage radius $R_k(h)$ is. It was shown in \cite{Alzenad} that for any QoS requirement $L_{\textup {th}}^k$, the optimal elevation angle $\theta^*$, that maximizes the coverage radius, is constant and depends only on the environment. The optimal elevation angle is given by\cite{Alzenad} 
\begin{equation}
\theta^*= \tan^{-1}(\frac{h^*_k}{R^*_k})\label{Eq:OptTheta}
\end{equation}
 where $h^*_k$ and $R^*_k$ are the optimal altitude that maximizes the coverage region and the associated maximum coverage radius, respectively, and optimal elevation angle $\theta^*=20.34^{\circ}, 42.44^{\circ}, 54.62^{\circ}$ and $75.52^{\circ}$ for the suburban, urban, dense urban and high-rise urban environments, respectively\cite{Alzenad}. For a given environment and a path loss threshold $L_{\textup{th}}^k$, the maximum coverage radius  can be evaluated by solving (\ref{Eq:PL3}). Finally, $h_k^*$ can be evaluated by solving (\ref{Eq:OptTheta}).
 
\section{Problem formulation and algorithms}
  As discussed previously, the coverage region for each set $\mathcal U_k$, denoted by $C_k$, is a circular disk with center $(x_D,y_D)$ and radius $R_k(h)$. Therefore, placing the coverage regions $\{C_k\}_{k=1}^K$ horizontally corresponds to placing the UAV-BS in the horizontal dimension. It is worth mentioning that the coverage regions $\{C_k\}_{k=1}^K$ have the same center which corresponds to the horizontal location of the UAV-BS, i.e., $(x_D,y_D)$. The user $i$ of set $\mathcal U_k$ is covered if it is located within a distance at most $R_k(h)$ from the center $(x_D,y_D)$. Let $u_{ik}\in \{0,1\}$ be a binary variable such that $u_{ik}=1$ if the user $i$ of set $\mathcal U_k$ is within the coverage region $C_k$ and $u_{ik}=0$ otherwise. This condition can be written as
\begin{equation}\label{Prb:Constraint1Before}
u_{ik}((x_{ik}-x_D)^2+(y_{ik}-y_D)^2)^\frac{1}{2}\leq R_k(h)
\end{equation}
Clearly, when $u_{ik}=1$,$((x_{ik}-x_D)^2+(y_{ik}-y_D)^2)^\frac{1}{2}\leq R_k(h)$ must be satisfied. On the other hand, when $u_{ik}=0$, the constraint (\ref{Prb:Constraint1Before}) is trivially satisfied. To avoid the multiplication of the variables $u_{ik}, x_D$ and $y_D$, we use the big-M method. The constraint (\ref{Prb:Constraint1Before}) can thus be further rewritten as 
\begin{equation}\label{Prb:ConstraintAfter}
((x_{ik}-x_D)^2+(y_{ik}-y_D)^2)^\frac{1}{2}\leq R_k(h)+M(1-u_{ik})
\end{equation}
where $M$ is a constant chosen large enough such that the constraint (\ref{Prb:ConstraintAfter}) is trivially satisfied when $u_{ik}=0$. The 3D placement problem can be formulated as
\begin{equation} \label{Prb:MainPrb}
\begin{aligned} 
&\underset{x_D,y_D,h,{u_{ik}}}{\operatorname{maximize}}  \hspace{0.2cm}   \sum_{k=1}^{K}\sum_{i\in \mathcal U_k}u_{ik} \\
&\text{subject to} \\ 
&((x_{ik}-x_D)^2+(y_{ik}-y_D)^2)^\frac{1}{2}\leq R_k(h)+M(1-u_{ik}),\\&\hspace{5cm} \forall i\in \mathcal U_k, k=1,2..K, \\
&u_{ik}\in \left\{0,1\right\}, \hspace{3cm} \forall i\in \mathcal U_k, k=1,2,..K. 
 \end{aligned}
\end{equation} 
The problem (\ref{Prb:MainPrb}) is a mixed integer non-linear problem (MINLP) which is difficult to solve. The difficulty of the problem (\ref{Prb:MainPrb}) arises due to the coupling between the vertical placement, i.e., $h$ and the horizontal placement, i.e., $(x_D,y_D)$ through the parameters $\{R_k(h)\}_{k=1}^K$. In order to simplify problem (\ref{Prb:MainPrb}), we decouple the vertical and the horizontal placements. Such decoupling can be performed by utilizing an exhaustive search for the optimal altitude that solves (\ref{Prb:MainPrb}).  
In the following lemma, we show that there exists a closed region in which the optimal altitude is guaranteed to exist.  

\begin{lemma}
Let $h_1^*$ and $R_1^*$ be the optimal altitude and the associated maximum coverage radius corresponding to the smallest path loss threshold $L_{th}^1$, respectively, and let $h_K^*$ and $R_K^*$ be the optimal altitude and the associated maximum coverage radius corresponding to the largest path loss threshold $L_{th}^K$, respectively, then $\exists h^* \in [h_1^*,h_K^*]$ that yields $N(h^*) \geq N(\bar{h}) \: \forall \: \bar{h}\not\in \: [h_1^*,h_K^*]$, where $N(h)$ is the number of covered users obtained by solving (\ref{Prb:MainPrb}) for a given $h$. 
\end{lemma}
\begin{proof}
Note that $R_k(h)$ is a concave function in $h$ and has one maxima at $h_k^*$ \cite{MozaffariDrone}. Also, we note that $h_m^*<h_l^*$ if $L_{th}^m < L_{th}^l$ since from (\ref{Eq:PL2}) we have $R_m(h)< R_l(h)$ if $ L_{th}^m < L_{th}^l$ and also  $\theta^*= \tan^{-1}(\frac{h_k^*}{R_k^*})$ is constant. Therefore, $\forall \bar{h} < h_1^*$ we have $R_k(\bar{h})<R_k(h_1^*)\: \forall k \overset{(a)}{\implies} N(\bar{h}) \leq N(h_1^*)$. Similarly, $\forall \bar{h} > h_K^*$ we have $R_k(h_K^*)>R_k(\bar{h}) \: \forall k \overset{(b)}{\implies} N(h_K^*) \geq N(\bar{h})$, which completes the proof. (a) and (b) result from the fact that increasing $R_k(h)$ enlarges the feasible region of (\ref{Prb:MainPrb}) which does not decrease the optimal value of the objective function of (\ref{Prb:MainPrb}).
\end{proof}

\subsection{Exhaustive search (ES):}

The ES algorithm performs an exhaustive search for the optimal altitude $h_{\textup {E}}^*$ over the closed region $[h_1^*,h_K^*]$. For a given altitude $h_E\in [h_1^*,h_K^*]$, the associated coverage radii $\{R_k(h_{\textup {E}})\}_{k=1}^K$ are computed by solving (\ref{Eq:PL2}) numerically. Next, (\ref{Prb:GS}) is solved to find the optimal horizontal UAV-BS location
\begin{equation} \label{Prb:GS}
\begin{aligned} 
&\underset{x_D,y_D,{u_{ik}}}{\operatorname{maximize}}  \hspace{0.2cm}   \sum_{k=1}^{K}\sum_{i\in \mathcal U_k}u_{ik} \\
&\text{subject to} \\ 
&((x_{ik}-x_D)^2+(y_{ik}-y_D)^2)^\frac{1}{2}\leq R_k(h_E)+M(1-u_{ik}),\\&\hspace{5cm} \forall i\in \mathcal U_k, k=1,2..K, \\
&u_{ik}\in \left\{0,1\right\}, \hspace{3cm} \forall i\in \mathcal U_k, k=1,2,..K. 
 \end{aligned}
\end{equation} 

The problem (\ref{Prb:GS}) is a mixed integer second order cone problem (MISOCP). Such problems can be solved by branch and cut method whose worst-case complexity is $O(2^n)$ where $n$ is the number of users. At each branching node, the underlying relaxed subproblem is a SOCP which can be solved in a polynomial time using primal-dual interior point method with complexity $O(n^{3.5} \log(\varepsilon ^{-1}))$ where $\varepsilon$ is the accepted duality gap \cite{nesterov1997self}. Thus, the complexity of solving problem (\ref{Prb:GS}) is $O(2^nn^{3.5} \log(\varepsilon ^{-1}))$.

\subsection{Maximal weighted area (MWA)} 
Let us consider the case in which the users in the set $\mathcal U_k$ are uniformly distributed over the serving region with density $\lambda_k$. Given that the UAV-BS is at altitude $h$, the average number of covered users, denoted by $N_{\textup {avg}}(h)$, is then
\begin{equation}\label{AvgNumUsers}
N_{\textup {avg}}(h)= \pi \sum_{k=1}^K {\lambda_k R_k^2(h)}.
\end{equation}

Obviously, maximizing the average number of covered users $N_{\textup {avg}}(h)$ for a uniformly distributed users depends only on the UAV-BS's altitude. In order to obtain the optimal altitude $h_M^*$ that maximizes (\ref{AvgNumUsers}), we need to search for $h$ that satisfies
\begin{equation}\label{Eq:FirsDerivative1}
\frac{\partial}{\partial h} {\sum_{k=1}^K {\lambda_k R_k^2(h)}}=0,
\end{equation}
which yields the following
\begin{equation}\label{Eq:FirsDerivative2}
\sum_{k=1}^K \frac{2 \lambda_k X_k(h) R_k^2(h)}{R_k^2(h)+h^2+hX_k(h)}=0,
\end{equation}
where 
\begin{equation}\label{Eq:Xk}
\resizebox{.88\columnwidth}{!}{$X_k(h)=\frac{-9 \ln(10)A a b}{\pi} \frac{R_k(h)  \exp (-b[\frac{180}{\pi}\tan^{-1}(\frac{h}{R_k(h)})-a])}{(1+a \exp (-b[\frac{180}{\pi}\tan^{-1}(\frac{h}{R_k(h)})-a]))^2}-h$}.
\end{equation}

It can be shown that (\ref{Eq:FirsDerivative2}) has a solution in the interval $[h_1^*,h_K^*]$. However, this solution may not be unique. Clearly, (\ref{Eq:FirsDerivative2}) is an implicit function of $h$. Therefore, we need to search for $h_M^*$ that satisfies (\ref{Eq:FirsDerivative2}) numerically. 

The maximal weighted area (MWA) algorithm deploys the UAV-BS at the altitude $h_{M}^*$\footnote{The objective function in (\ref{Prb:GS}) and (\ref{Prb:MA}) is to maximize the number of covered users (not the average number of covered users) by also optimizing the UAV-BS's horizontal location which clearly results in covering more users in comparison to placing the UAV-BS randomly in the horizontal dimension. The uniform distribution is assumed in the MWA algorithm to avoid the exhaustive search on the altitude by deploying the UAV-BS at the altitude $h_{M}^*$ .}. Let $\{R_k(h_{M}^*)\}_{k=1}^K$ be the coverage radii associated with the altitude $h_{M}^*$. The problem (\ref{Prb:MainPrb}) then reduces to

\begin{equation} \label{Prb:MA}
\begin{aligned} 
&\underset{x_D,y_D,{u_{ik}}}{\operatorname{maximize}}  \hspace{0.2cm}   \sum_{k=1}^{N}\sum_{i\in \mathcal U_k}u_{ik} \\
&\text{subject to} \\ 
&((x_{ik}-x_D)^2+(y_{ik}-y_D)^2)^\frac{1}{2}\leq R_k(h_M^*)+M(1-u_{ik}),\\&\hspace{5cm} \forall i\in \mathcal U_k, k=1,2..K, \\
&u_{ik}\in \left\{0,1\right\}, \hspace{3cm} \forall i\in \mathcal U_k, k=1,2,..K, 
 \end{aligned}
\end{equation} 
which is a MISOCP.

\section{simulation results} 
  
We consider a square 3 km $\times$ 3 km urban area with parameters $a=9.61$, $b=0.16$, $\eta_{\textup{LoS}}=1$  and $\eta_{\textup{NLoS}}=20$. We also consider a UAV-BS that transmits its signal at $f_c= 2$ GHz and  $P_t=30$ dBm. We assume that there are two sets of users $\mathcal U_1$ and $\mathcal U_2$ uniformly distributed with densities $\lambda_1$ and $\lambda_2$, respectively. However, for a fair comparison, the total density of users is fixed at $\lambda=\lambda_1+\lambda_2=11$ users/$\textup {km}^2$. Furthermore, we assume that the users demand QoS defined as $\gamma_\textup {th}^1= 50$ dBm and $\gamma_\textup {th}^2= 47$ dBm for $\mathcal U_1$ and $\mathcal U_2$, respectively, with $P_n=-120$ dBm. For comparison, we assume a UAV-BS placement algorithm, namely largest QoS (LQ) algorithm. The LQ algorithm assumes that all the users have the same QoS requirement $\gamma_\textup {th}= 50$ dB, i.e., $L_{t\textup {h}}=100$ dB and the UAV-BS is therefore deployed vertically at $h_{LQ}^*=646.5$ m which results in maximal coverage radius $R^*=707$ m. The LQ algorithm is based on the observation that any user $i$ of set $\mathcal U_k$, regardless of the required QoS, falling within the coverage region that corresponds to the largest SNR threshold will be covered. 
 For the ES algorithm, we perform an exhaustive search for the optimal altitude $h_{E}^*$ over the closed region $[646.5 , 913]$ m. Furthermore, we discretize the altitude range $[646.5,913]$ m into a uniform one-dimensional grid of 9 points where the discretization step is given by $\Delta h=29.6$ m. In this letter, we use the CVX parser/solver with the MOSEK solver to solve problems (12) and (17).
\begin{figure}
\begin{center}
\captionsetup{justification=centering}
\includegraphics[ height=5.5cm, width=8cm]{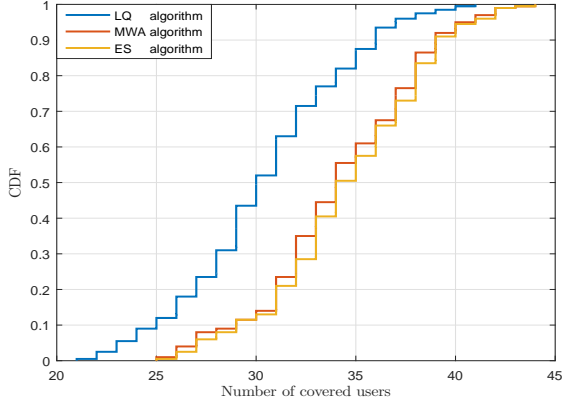}
\caption{\footnotesize CDF of the number of covered users ($\rho =1$).}
\label{fig2}
\end{center}
\end{figure}

The number of covered users and execution time are random quantities whose distributions can be measured by the cumulative distribution function (CDF). Fig. \ref{fig2} and Fig. \ref{fig3} show the CDF of the number of covered users and the CDF of the execution time for $\rho=\frac{\lambda_2}{\lambda_1}=1$, respectively. As shown in Fig. \ref{fig2}, the ES and MWA algorithms have very close performance and both outperform the LQ algorithm. However, based on Fig. \ref{fig3}, the ES algorithm has the worst execution time with a significant gap to that of the MWA and the LQ algorithms.

 
 Fig.~\ref{fig4} shows the average number of covered users versus the density ratio $\rho$. Clearly, the performance of the MWA algorithm is very close to that of the ES algorithm. It is also worth noting that as $\rho$ increases, the gap between the MWA and ES algorithms on one hand and the LQ algorithm on the other hand increases. This is because as $\rho$ increases, the number of elements in $\mathcal U_2$ ($\mathcal U_1$) increases (decreases). However, the LQ algorithm does not consider the density of the users in the set $\mathcal U_2$ which justifies the gap increase.


\section{Conclusion}
In this letter, we studied a novel 3D placement of a UAV-BS that maximizes the number of covered users with different QoS requirements. We modeled the placement problem as a multiple circles placement problem and proposed an optimal placement algorithm that utilizes an exhaustive search over a one-dimensional parameter in a closed region.  We also proposed a low-complexity algorithm, referred to as the MWA algorithm, to solve the placement problem. Simulations have shown that the MWA algorithm performs very close to the ES algorithm with a significant reduction in complexity.   


\begin{figure}
\begin{center}
\captionsetup{justification=centering}
\includegraphics[ height=5.5cm, width=8cm]{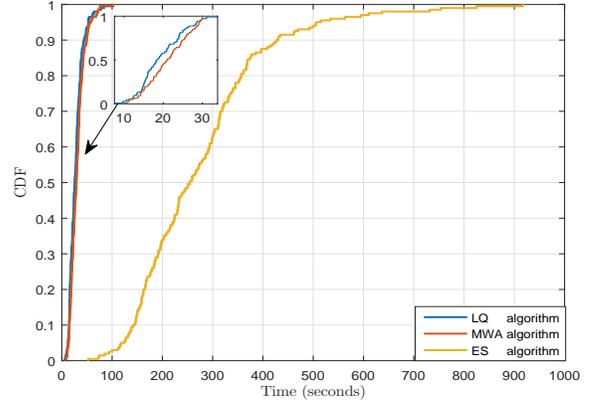}
\caption{\footnotesize CDF of execution time ($\rho =1$).}
\label{fig3}
\end{center}
\end{figure} 

 \begin{figure}
\begin{center}
\captionsetup{justification=centering}
\includegraphics[ height=5.5cm, width=8cm]{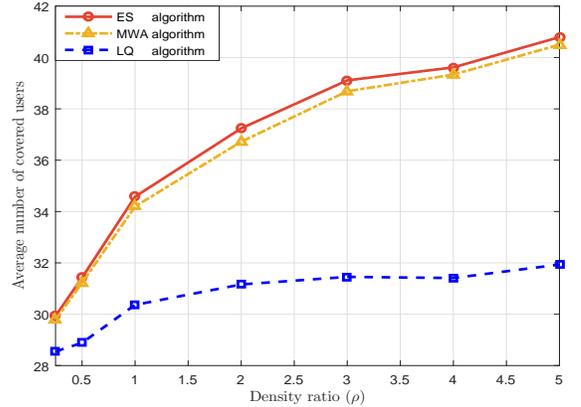}
\caption{\footnotesize Average number of covered users versus density ratio.}
\label{fig4}
\end{center}
\end{figure} 
\bibliographystyle{IEEEtran}
\bibliography{IEEEfull,UAV-BS_Placement}

\begin{thebibliography}{10}
\providecommand{\url}[1]{#1}
\csname url@samestyle\endcsname
\providecommand{\newblock}{\relax}
\providecommand{\bibinfo}[2]{#2}
\providecommand{\BIBentrySTDinterwordspacing}{\spaceskip=0pt\relax}
\providecommand{\BIBentryALTinterwordstretchfactor}{4}
\providecommand{\BIBentryALTinterwordspacing}{\spaceskip=\fontdimen2\font plus
\BIBentryALTinterwordstretchfactor\fontdimen3\font minus
  \fontdimen4\font\relax}
\providecommand{\BIBforeignlanguage}[2]{{%
\expandafter\ifx\csname l@#1\endcsname\relax
\typeout{** WARNING: IEEEtran.bst: No hyphenation pattern has been}%
\typeout{** loaded for the language `#1'. Using the pattern for}%
\typeout{** the default language instead.}%
\else
\language=\csname l@#1\endcsname
\fi
#2}}
\providecommand{\BIBdecl}{\relax}
\BIBdecl

\bibitem{IremMagazine}
I.~Bor-Yaliniz and H.~Yanikomeroglu, ``The new frontier in {RAN} heterogeneity:
  Multi-tier drone-cells,'' \emph{IEEE Commun. Mag.}, vol.~54, no.~11, pp.
  48--55, Nov. 2016.

\bibitem{zeng2016wireless}
Y.~Zeng, R.~Zhang, and T.~J. Lim, ``Wireless communications with unmanned
  aerial vehicles: Opportunities and challenges,'' \emph{IEEE Commun. Mag.},
  vol.~54, no.~5, pp. 36--42, May 2016.

\bibitem{mozaffari2017mobile}
M.~Mozaffari, W.~Saad, M.~Bennis, and M.~Debbah, ``Mobile unmanned aerial
  vehicles {(UAVs)} for energy-efficient internet of things communications,''
  \emph{arXiv preprint arXiv:1703.05401}, 2017.

\bibitem{zeng2017energy}
Y.~Zeng and R.~Zhang, ``Energy-{E}fficient {UAV} communication with trajectory
  optimization,'' \emph{IEEE Trans. Wireless Commun.}, vol.~16, no.~6, pp.
  3747--3760, Jun. 2017.

\bibitem{lyu2017placement}
J.~Lyu, Y.~Zeng, R.~Zhang, and T.~J. Lim, ``Placement optimization of
  {UAV}-mounted mobile base stations,'' \emph{IEEE Commu. Lett.}, vol.~21,
  no.~3, pp. 604--607, Mar. 2017.

\bibitem{Alzenad}
M.~Alzenad, A.~El-keyi, F.~Lagum, and H.~Yanikomeroglu, ``3{D} placement of an
  unmanned aerial vehicle base station ({UAV-BS}) for energy-efficient maximal
  coverage,'' \emph{to appear in IEEE Wireless Commun. Lett.}, pp. 1--4, 2017,
  DOI:10.1109/LWC.2017.2700840.

\bibitem{Hourani}
A.~Al-Hourani, S.~Kandeepan, and S.~Lardner, ``Optimal {LAP} altitude for
  maximum coverage,'' \emph{IEEE Wireless Commun. Lett.}, vol.~3, no.~6, pp.
  569--572, Dec. 2014.

\bibitem{IremConf}
I.~Bor-Yaliniz, A.~El-Keyi, and H.~Yanikomeroglu, ``Efficient 3-{D} placement
  of an aerial base station in next generation cellular networks,''
  \emph{\textup{in} Proc. IEEE Int. Conf. Commun. (ICC)}, pp. 1--5, Kuala
  Lumpur, Malaysia, May 2016.

\bibitem{Elham2017Backhaul}
E.~Kalantari, M.~Z. Shakir, H.~Yanikomeroglu, and A.~Yongacoglu,
  ``Backhaul-aware robust 3{D} drone placement in 5{G}+ wireless networks,''
  \emph{\textup{in} Proc. IEEE Int. Conf. Commun. Workshop (ICCW)}, Paris,
  France, May 2017.

\bibitem{ElhamVtc}
E.~Kalantari, H.~Yanikomeroglu, and A.~Yongacoglu, ``On the number and {3D}
  placement of drone base stations in wireless cellular networks,'' in
  \emph{2016 IEEE 84th Veh. Technol. Conf. (VTC Fall)}, Montreal, Canada, Sep.
  2016.

\bibitem{MozaffariDrone}
M.~Mozaffari, W.~Saad, M.~Bennis, and M.~Debbah, ``Drone small cells in the
  clouds: Design, deployment and performance analysis,'' \emph{\textup{in}
  Proc. 2015 IEEE Global Commun. Conf. (GLOBECOM)}, pp. 1--6, San Diego, USA,
  Dec. 2015.

\bibitem{nesterov1997self}
Y.~E. Nesterov and M.~J. Todd, ``Self-scaled barriers and interior-point
  methods for convex programming,'' \emph{Mathematics of Operations research},
  vol.~22, no.~1, pp. 1--42, 1997.

\end{thebibliography}
\end{document}